\documentclass[journal]{IEEEtran}
\usepackage[utf8]{inputenc}
\usepackage{amsmath}
\usepackage{amssymb}
\usepackage{algorithm}
\usepackage{algpseudocode}
\usepackage{array}
\usepackage{amsthm,color}
\usepackage{amsmath}
\usepackage{multirow}
\usepackage[style=base]{caption}
\usepackage{subcaption}
\usepackage{graphicx}
\usepackage{epstopdf}

\usepackage{url}
\usepackage{tikz}
\usepackage{makecell}
\usepackage{mathrsfs}
\usepackage{tikz}
\usepackage{color}
\usepackage{amsmath}
\usepackage{amssymb}
\usepackage{amsthm}
\usepackage{graphicx}
\usepackage{extarrows}
\usepackage{times}
\usepackage{graphics,color}
\newtheorem{theorem}{Theorem}

\newtheorem{corollary}{Corollary}
\newtheorem{definition}{Definition}

\newtheorem{myassum}{Assumption}

\newcommand{\R}{\mathbb R}

\graphicspath{{figs/}}

\usepackage[square,numbers]{natbib}
\newcommand{\ignore}[1]{}

\begin{document}
\title{
Synchronization analysis of high order layered complex networks}

\author{~Yujuan~Han, Wenlian~Lu, and Tianping~Chen
\thanks{This work was supported by the National Key $R\&D$ Program of China (No. 2018AAA0100303), the National Natural Science Foundation of China (No. 62072111,61703271, 51879156), the Shanghai Municipal Science and Technology Major Project (No. 2018SHZDZX01) and the ZHANGJIANG LAB, and the Technology Commission of Shanghai Municipality (No. 19JC1420101).}
\thanks{T. Chen is with
University, Shanghai 200433, China. W. Lu are with the School of Mathematical Sciences, Fudan
University, Shanghai 200433, China, also with the Shanghai Center for
Mathematical Sciences, Fudan University, Shanghai 200433, China, and also
with the Shanghai Key Laboratory for Contemporary Applied Mathematics,
Fudan University, Shanghai 200433, China. Y. Han is with College of Information Engineering, Shanghai Maritime University, Shanghai 201306, China. Corresponding author: Tianping Chen. Email: tchen@fudan.edu.cn}
}

\maketitle

\begin{abstract}
In this letter, we propose high order layered complex networks. The synchronization is discussed in detail. The relations of synchronization, individual coupling matrices and the intrinsic function of the uncoupled system are given. As special cases, synchronization of monolayer networks and multiplex networks discussed in literature can be obtained.
\end{abstract}


\section{Introduction}

In recent years, synchronization of linearly coupled networks has been an important topic
in the research of complex networks. Many models and mathematics concepts have been proposed \cite{Pecora1990}-\cite{Sevilla2016Interlayer}.
For example, in \cite{Chen06New}, \cite{Lu2021QUAD}, the following model was discussed
\begin{align}\label{model1}
\frac{d x_{i}(t)}{dt}=f(x_{i}(t))
+c\sum\limits_{j=1}^{m}l_{ij} x_{j}(t),\quad i=1,\cdots,N
\end{align}
where $x_{i}(t)\in R^{n}$ is the state variable of the $i$th
node, $f:R^{n}\rightarrow R^{n}$,
$L=(l_{ij})\in R^{N\times N}$ is the coupling matrix with zero-sum
rows and $l_{ij}\ge 0$, for  $i\ne j$. The roles played by the intrinsic property of $f(x)$ and the coupling matrix $L$ in synchronization was discussed in \cite{Chen06New}.

In practice, many networks are interconnected, forming a network of networks, which can be viewed as multi-layer networks \cite{Gao2012Networks}-\cite{Wu2020Impact}. Multiplex networks is one kind of multi-layer networks that all
layers have the same nodes and nodes from one layer only interact with their counterparts in other layers.
In \cite{Sevilla2016Interlayer}-\cite{Xu2015Synchronizability}, synchronization of multiplex networks was analyzed. In \cite{Deng2020Eigenvalue}, \cite{Khalaf2019Synchronization}, the following model was studied
\begin{align}\label{model2}
\frac{d x_{i}^{k}(t)}{dt}=& f(x_{i}^{k}(t))+c\sum_{j=1}^{N}
a_{ij} x_{j}^{k}(t)+d\sum_{l=1}^{K}
b_{kl} x_{i}^{l}(t),\nonumber\\
&~~i=1,\cdots,N,~k=1,\cdots,K
\end{align}
where $x_i^k(t)\in\R^n$ is the state of the $i$th node in the $k$th layer, $A=(a_{ij})\in R^{N\times N}$ is the intra-layer coupling matrix within layers with zero-sum rows and $a_{ij}\ge 0$, for  $i\ne j$, $B=(b_{kl})\in R^{K\times K}$ represents the inter-layer coupling matrix with zero-sum rows and $b_{kl}\ge 0$, for  $k\ne l$, $c$ and $d$ are coupling strengths.

For multiplex networks, two forms of synchronization phenomenon can emerge, i.e., intra-layer and inter-layer synchronization \cite{Bidesh2019Intralayer}. Intra-layer synchronization means that all nodes within the same layer evolve synchronously, while inter-layer synchronization means that each node in one layer evolves synchronously with all its counterparts in other layers.

In this paper, we propose more general multi-layer networks, called high order layered networks, and study synchronization of these networks. More precisely, we investigate the following high order layered coupled networks
\begin{align}\label{model}
\frac{d x_{i_{1},\cdots,i_{m}}(t)}{dt}&=f(x_{i_{1},\cdots,i_{m}}(t))
+c_{1}\sum\limits_{j=1}^{I_{1}}l^{(1)}_{i_{1}j} x_{j,i_{2},\cdots,i_{m}}(t)
\nonumber\\&
+c_{2}\sum\limits_{j=1}^{I_{2}}l^{(2)}_{i_{2}j} x_{i_{1},j,\cdots,i_{m}}(t)+\cdots
\nonumber\\&
+c_{k}\sum\limits_{j=1}^{I_{k}}l^{(k)}_{i_{k}j} x_{i_{1},\cdots,i_{k-1},j,i_{k+1},\cdots,i_{m}}(t)+\cdots\nonumber\\&
+c_{m}\sum\limits_{j=1}^{I_{m}}l^{(m)}_{i_{m}j} x_{i_{1},\cdots,j
}(t)
\end{align}

It is clear that models (\ref{model1}) and (\ref{model2}) are special cases of the general model (\ref{model}). Therefore, the results on model (\ref{model}) can apply to these special cases.

In this paper, we study synchronization of high order layered networks (\ref{model}) and explore how the synchronization depends on the eigenvalues of the coupling matrices $L^{(k)}$ and the coupling strengths $c_{k}$, $k=1,\cdots,m$.

This paper is organized as follows. Section \ref{sec-model} introduces the models and give some basic notations, definitions and assumptions. Section \ref{sec-main} investigates synchronization of high order layered networks in detail. Numerical examples are given in Section \ref{sec-simu} to illustrate the effectiveness of the theoretical results. Finally, Section \ref{sec-conc} concludes the paper.

\section{Model description and background}\label{sec-model}

Let us consider a setup in which the dynamics of nodes in high order layered networks is linearly coupled. The coupling among nodes are assumed to be undirected.

We call complex network (\ref{model1}) the {\em first order layered network} and multiplex network (\ref{model2}) the {\em second order layered network}.

Notice that each layer of multiplex networks is a first order layered network. If we take each layer as a ``node" whose own dynamics follows model (\ref{model1}) and $b_{kl}$ as the coupling strength between the $k$th and $l$th ``nodes", these linearly coupled ``nodes" will generate the multiplex network (\ref{model2}). This implies that the second order layered network is in fact a network composed of identical first order layered networks.


Inductively, we can obtain a {\em $m$th order layered network} by regarding the $(m-1)$th order layered network as a ``node" and linearly coupling these ``nodes". And we call the network topology among the $(m-1)$th order layered networks the {\em $m$th order subnetwork}. In sequence, we can define the {\em $k$th order subnetwork} of the $m$th order layered network, $k=m-1,\cdots,1$ (see Fig. \ref{fig-network}).

Suppose the $m$th order layered network have $I_k$ nodes in its $k$th order subnetwork, $k=1,\cdots,m$. Denote by $L^{(k)}\in\R^{I_k\times I_k}$ the coupling matrix of the $k$th order subnetwork with zero-sum rows, i.e., $[L^{(k)}]_{ij}\ge 0$ for $i\neq j$ and $[L^{(k)}]_{ii} = -\sum_{j=1}^{I_k}[L^{(k)}]_{ij}$, and $c_k$ the coupling strength among nodes in the $k$th order subnetwork.

Next, we will introduce the model of high order layered network inductively. Firstly, we consider the first order layered network. Denote by $x_{i_1}^{(1)}(t)$ the state of the $i_1$th node in the first order layered network. Then the dynamics of the first order layered network satisfies
\begin{eqnarray}\label{one-layer-model}
\frac{d x_{i_1}^{(1)}(t)}{dt}=f(x_{i_{1}}^{(1)}(t))+c_1\sum_{j=1}^{I_{1}}
l_{i_{1},j}^{(1)} x_{j}^{(1)}(t),~i_1=1,\cdots,I_1.
\end{eqnarray}

Denote $x^{(1)}(t)=[x^{(1)}_{1}(t)^{T},\cdots,x^{(1)}_{I_{1}}(t)^{T}]^{T}$,
$f(x^{(1)}(t))=[f(x^{(1)}_{1}(t))^{T},\cdots,
f(x^{(1)}_{I_{1}}(t))^{T}]^{T}$. System (\ref{one-layer-model}) can be written in the matrix form
\begin{align}\label{matrix-form-model1}
\frac{d x^{(1)}(t)}{dt}=f(x^{(1)}(t))+c_{1}(L^{(1)}\otimes E_{n})x^{(1)}(t)
\end{align}
where $\otimes$ is the Kronecker product and $E_n$ is the $n\times n$ identity matrix. Consider that the second order layered network is obtained by linearly coupling $I_2$ identical first order layered networks. Denote by ${x}^{(2)}_{i_{2}}(t)$ the state of the $i_2$th first order layered network, whose uncoupled dynamics follows model (\ref{matrix-form-model1}). Hence, the dynamics of the second order layered network satisfies
\begin{align*}
\frac{d x_{i_2}^{(2)}(t)}{dt}
= & f(x^{(2)}_{i_{2}}(t)) +c_{1}(L^{(1)}\otimes E_{n}) x^{(2)}_{i_{2}}(t)\\
& +c_{2}\sum_{j=1}^{I_2} l_{i_{2}j}^{(2)} x_{j}^{(2)}(t)
\end{align*}
Denote $x^{(2)}(t)=[x^{(2)}_{1}(t)^{T},\cdots,x^{(2)}_{I_{2}}(t)^{T}]^{T}$,
$f(x^{(2)}(t))=[f(x^{(2)}_{1}(t))^{T},\cdots,
f(x^{(2)}_{I_{2}}(t))^{T}]^{T}$. Then the dynamics of second order layered networks can be written in the matrix form
\begin{align}
\frac{d x^{(2)}(t)}{dt}&=f(x^{(2)}(t))
+c_{1}(E_{I_2}\otimes L^{(1)}\otimes E_{n}) x^{(2)}(t)\nonumber\\
&+c_{2}(L^{(2)}\otimes E_{I_{1}}\otimes E_n) x^{(2)}(t)
\end{align}

Inductively, we obtain the dynamics of the $m$th order layered network as follows
\begin{align}\nonumber
&\frac{d x^{(m)}(t)}{dt}=f(x^{(m)}(t))
+\sum_{k=1}^{m}c_{k}(
E_{I_{m}}\otimes \cdots \otimes E_{I_{k+1}}\otimes
L^{(k)}\\\label{high-order-model}
& ~~~~~~~~~~~~~~~~~~~~~~~~~\otimes E_{I_{k-1}}\cdots \otimes E_{I_{1}}\otimes E_n)x^{(m)}(t)
\end{align}	
where $x^{(m)}(t)=[x^{(m)}_{1}(t)^{T},\cdots,x^{(m)}_{I_{m}}(t)^{T}]^{T}\in\R^{nI_1\cdots I_m}$.

Notice that the $m$th order layered network has a total number of $(I_m\cdots I_2I_1)$ nodes in the network (see red nodes in Fig. \ref{fig-network}). Let us consider the dynamics of these node, which are denoted by $v_{i_1,i_2,\cdots,i_m}$ with $i_k$ being the node index in the $k$th order subnetwork, $k=1,\cdots,m$. The state of each node $v_{i_1,i_2,\cdots,i_m}$ is represented as $x_{i_1,i_2,\cdots,i_m}\in\R^n$. Apparently, $x_{i_1,i_2,\cdots,i_m} = [\cdots[[x^{(m)}_{i_m}]_{i_{m-1}}]_{i_{m-2}}\cdots]_{i_1}$. Then from (\ref{high-order-model}), we get that the dynamics of node $v_{i_1,i_2,\cdots,i_m}$ in the $m$th order layered network satisfies model (\ref{model}).

Similar to synchronization in multiplex networks, different forms of synchronization phenomenon can emerge in high order layered networks. 

In this paper, synchronization of high order layered networks is defined as follows.
\begin{definition}
Consider the $m$th order layered network (\ref{model}).
\begin{enumerate}
  \item For any $1\le k\le m$, network (\ref{model}) is said to reach $k$th order synchronization if for any $1\le p,q\le I_k$,
\begin{align}\nonumber
&\lim_{t\to\infty} \|x_{i_1,\cdots,i_{k-1},p,i_{k+1},\cdots,i_m}(t)-x_{i_1,\cdots,i_{k-1},q,i_{k+1},\cdots,i_m}(t)\|\\
&=0
\end{align}
holds for any fixed $i_s=1,\cdots,I_s, s=1,\cdots,k-1,k,\cdots,m$.
  \item Network (\ref{model}) is said to realize complete synchronization if
\begin{align}
\lim_{t\to\infty} \|x_{i_1,i_2,\cdots,i_m}(t)-x_{j_1,j_2,\cdots,j_m}(t)\|=0
\end{align}
holds for any $1\le i_k,j_k\le I_k$, $k=1,\cdots,m$,
\end{enumerate}
\end{definition}

Throughout this paper, we suppose $f$ satisfies the following QUAD condition. The role of the QUAD condition in synchronization of complex networks was investigated in \cite{Chen06New,Lu2021QUAD,Delellis2011On}.
\begin{myassum}\label{assump-f}
The continuous function $f(x,t):\mathbb R^n\times [0,+\infty)\to\mathbb R^n$ is said to satisfy QUAD condition if there exists a constant $\alpha$ such that
\begin{align}
(x-y)^{\top}\left[f(x)-f(y)\right]
\le \alpha(x-y)^{\top}(x-y)
\end{align}
holds for all $x, y\in\R^{n}$.
\end{myassum}
\begin{myassum}\label{assump-L}
All coupling matrices $L^{(k)} = [l^{(k)}_{i,j}]\in\R^{I_{k}\times I_k}$, $k=1,\cdots,m$ are symmetric.  For any $k$, the eigenvalues of $L^{(k)}$ are $0= \lambda_{1}^{(k)}\ge \lambda_{2}^{(k)}\ge \cdots\ge\lambda_{I_k}^{(k)}$.
\end{myassum}

\section{Main Results}\label{sec-main}

In this section, we provide a unified approach to analyse synchronization of the network (\ref{model}). The following theorem reveals how the synchronization depends on the eigenvalues of the coupling matrices $L^{(k)}$ and the coupling strengths $c_{k}$, $k=1,\cdots,m$.

\begin{theorem}\label{thm-1}
Suppose Assumptions \ref{assump-f} and \ref{assump-L} are satisfied. For any $k=1,\cdots,m$, the $m$th order layered network (\ref{model}) can reach $k$th order synchronization if $c_k\lambda_2^{(k)}+\alpha<0$, where $\alpha$ is the constant defined in Assumption 1.
\end{theorem}
\begin{proof}
Here we mainly consider the case  $m=3$. The proof for case $m>3$ can be given in a similar argument. For any fixed $1\le i_2\le I_2$ and $1\le i_3\le I_3$, we define
\begin{align}
\bar{x}^{(1)}_{\cdot,i_2,i_3}(t)=\frac{1}{I_{1}}\sum_{j=1}^{I_{1}}x_{j,i_2,i_3}(t).
\end{align}
Define a Lyapunov funciton
\begin{align}
V(X(t))=\frac{1}{2}\sum_{i_1,i_2,i_3}\|x_{i_1,i_2,i_3}(t)-\bar{x}^{(1)}_{\cdot,i_2,i_3}(t)\|_2^2.
\end{align}
Because of $\sum_{i_1=1}^{I_{1}}(x_{i_1,i_2,i_3}(t)-\bar{x}^{(1)}_{\cdot,i_2,i_3}(t))=0$, we have
\begin{align}\label{dot-barx}
\sum_{i_1,i_2,i_3}[x_{i_1,i_2,i_3}(t)-\bar{x}^{(1)}_{\cdot,i_2,i_3}(t)]^{\top}\frac{d\bar{x}^{(1)}_{\cdot,i_2,i_3}(t)}{dt}=0
\end{align}
and
\begin{align}\label{V1-barf}
\sum_{i_1,i_2,i_3}[x_{i_1,i_2,i_3}(t)-\bar{x}^{(1)}_{\cdot,i_2,i_3}(t)]^{\top}f(\bar{x}^{(1)}_{\cdot,i_2,i_3}(t))=0.
\end{align}
Differentiating $V(X(t))$ and by (\ref{dot-barx}), we have
\begin{align}\nonumber
&\frac{dV(X(t))}{dt}=\sum_{i_1,i_2,i_3}[x_{i_1,i_2,i_3}(t)-\bar{x}^{(1)}_{\cdot,i_2,i_3}(t)]^{\top}f(x_{i_1,i_2,i_3}(t))\\\nonumber
&~~~~~~+c_1\sum_{i_1,i_2,i_3}\sum_{j=1}^{I_1}[x_{i_1,i_2,i_3}(t)-\bar{x}^{(1)}_{\cdot,i_2,i_3}(t)]^{\top}l_{i_1j}^{(1)}x_{j,i_2,i_3}(t)\\\nonumber
&~~~~~~+c_2\sum_{i_1,i_2,i_3}\sum_{j=1}^{I_2}[x_{i_1,i_2,i_3}(t)-\bar{x}^{(1)}_{\cdot,i_2,i_3}(t)]^{\top}l_{i_2j}^{(2)}x_{i_1,j,i_3}(t)\\\nonumber
&~~~~~~+c_3\sum_{i_1,i_2,i_3}\sum_{j=1}^{I_3}[x_{i_1,i_2,i_3}(t)-\bar{x}^{(1)}_{\cdot,i_2,i_3}(t)]^{\top}l_{i_3j}^{(3)}x_{i_1,i_2,j}(t)\\
&~~~~= V_1(t) + V_2(t) + V_3(t)+V_4(t)
\end{align}
From (\ref{V1-barf}) and Assumption 1, we have
\begin{align}
V_1(t)\le \alpha\sum_{i_1,i_2,i_3}\|x_{i_1,i_2,i_3}(t)-\bar{x}^{(1)}_{\cdot,i_2,i_3}(t)\|^2_2
\end{align}
By Assumption 2 that $L^{(1)}$ is symmetric and has all row sums zero, we have
\begin{align}\nonumber
V_2(t) &= c_1\sum_{i_1,i_2,i_3}\sum_{j=1}^{I_1}l_{i_1 j}^{(1)}[x_{i_1,i_2,i_3}(t)-\bar{x}^{(1)}_{\cdot,i_2,i_3}(t)]^{\top}\\\nonumber
&~~~~~~~~~~~~~~~~~~~~~\cdot[x_{j,i_2,i_3}(t)-\bar{x}^{(1)}_{\cdot,i_2,i_3}(t)]\\
&\le c_1\lambda_2^{(1)}\sum_{i_1,i_2,i_3}\|x_{i_1,i_2,i_3}(t)-\bar{x}^{(1)}_{\cdot,i_2,i_3}(t)\|_2^2
\end{align}
By Assumption 2 that $L^{(2)}$ is symmetric, semi-negative and has all row sums zero, we have
\begin{align}
V_3(t) &= c_2\sum_{i_1,i_2,i_3}\sum_{j=1}^{I_2}l_{i_2 j}^{(2)}x_{i_1,i_2,i_3}(t)^{\top}x_{i_1,j,i_3}(t)
\le 0
\end{align}
Similarly, we can derive $V_4(t)\le 0$. Therefore,
\begin{align}
\frac{dV(X(t))}{dt}\le 2(c_1\lambda_2^{(1)}+\alpha)V(t)
\end{align}
which implies for any fixed $1\le i_2\le I_2$ and $1\le i_3\le I_3$,
\begin{align}\label{1th-syn}
\|x_{p,i_2,i_3}(t)-x_{q,i_2,i_3}(t)\|=O\left(e^{2(c_1\lambda_2^{(1)}+\alpha)t}\right),~1\le p,q\le I_1
\end{align}
i.e. the 1st order synchronization can be realized if $c_1\lambda_2^{(1)}+\alpha<0$.
Similarly, we can derive that for any fixed $1\le i_1\le I_1$ and $1\le i_3\le I_3$,
\begin{align}\label{2nd-syn}
\|x_{i_1,p,i_3}(t)-x_{i_1,q,i_3}(t)\|=O\left(e^{2(c_2\lambda_2^{(2)}+\alpha)t}\right),~1\le p,q\le I_2
\end{align}
and for any fixed $1\le i_1\le I_1$ and $1\le i_2\le I_2$,
\begin{align}\label{3rd-syn}
\|x_{i_1,i_2,p}(t)-x_{i_1,i_2,q}(t)\|=O\left(e^{2(c_3\lambda_2^{(3)}+\alpha)t}\right),~1\le p,q\le I_3.
\end{align}
From (\ref{1th-syn}), (\ref{2nd-syn}) and (\ref{3rd-syn}), we have that the $k$th order synchronization can be realized if $c_k\lambda_2^{(k)}+\alpha<0$, $k=1,2,3$. This completes the proof.
\end{proof}
Combing (\ref{1th-syn}), (\ref{2nd-syn}) and (\ref{3rd-syn}), we have that for any $1\le i_k,j_k\le I_k$, $k=1,2,3$,
\begin{align}
\|x_{i_1,i_2,i_3}(t)-x_{j_1,j_2,j_3}(t)\|=O\left(e^{\max_{k}2(c_k\lambda_2^{(k)}+\alpha)t}\right)
\end{align}
This implies that complete synchronization of the 3rd order layered network can be realized if $\max_{k=1,2,3}(c_k\lambda_2^{(k)}+\alpha)<0$, which leads to the following corollary.
\begin{corollary}
Suppose Assumptions 1 and 2 are satisfied. Then the $m$th order layered network (\ref{model}) will realize complete synchronization if $\max_{k=1,\cdots,m}(c_k\lambda_2^{(k)}+\alpha)<0$.
\end{corollary}

It is clear that the multiplex network (\ref{model2}) is a special case of model (\ref{model}) with $m=2$, $L^{(1)}=A$ and $L^{(2)}=B$. Therefore, by Theorem \ref{thm-1}, we obtain the following corollary.
\begin{corollary}
Consider the multiplex network (\ref{model2}). Suppose Assumption 1 is satisfied and coupling matrices $A$ and $B$ are symmetric. The eigenvalues of $A$ are $0=\mu_1\ge\mu_2\ge\cdots\ge\mu_N$. The eigenvalues of $B$ are $0=\nu_1\ge\nu_2\ge\cdots\ge\nu_K$. Then the multiplex network will realize i) intra-layer synchronization if $c\mu_2+\alpha<0$, ii) inter-layer synchronization if $d\nu_2+\alpha<0$, and iii) complete synchronization if $c\max\{\mu_2,\nu_2\}+\alpha<0$.
\end{corollary}

\section{Numerical simulations}\label{sec-simu}
In this section, we give some numerical examples to verify the effectiveness of theoretical results given in previous.

In these examples, the network is a 3rd order layered network whose full graph is given in Fig. \ref{fig-network}(a).
Denote by $x_{i_1,i_2,i_3}(t)$ the state of node $v_{i_1,i_2,i_3}$ at time $t$, where $1\le i_1\le 4$, $1\le i_2\le 3$, $1\le i_3\le 2$. The intrinsic dynamics of each node is the Lorenz system described by
\begin{align}
\dot{s}(t)=f(s(t))=\left(
\begin{array}{c}
10(s^2(t)-s^1(t))\\
28s^1(t)-s^2(t)-s^1(t)s^3(t)\\
s^1(t)s^2(t)-8s^3(t)/3
\end{array}\right).
\end{align}

Define qualities $E_{1}(t)$, $E_{2}(t)$ and $E_3(t)$ for the $1$st order, $2$nd order and $3$rd order synchronization errors respectively, where
\begin{align}\nonumber
&E_1(t) = \sum_{p<q}\sum_{i_2=1}^{3}\sum_{i_3=1}^2
\|x_{p,i_2,i_3}(t)- x_{q,i_2,i_3}(t)\|_2^2,\\\nonumber
&E_2(t) =  \sum_{i_1=1}^{4}\sum_{p<q}\sum_{i_3=1}^2
\|x_{i_1,p,i_3}(t)- x_{i_1,q,i_3}(t)\|_2^2,\\
&E_3(t) = \sum_{i_1=1}^{4}\sum_{i_2=1}^3
\|x_{i_1,i_2,1}(t)- x_{i_1,i_2,2}(t)\|_2^2.
\end{align}
As for the coupling matrices, we pick
\begin{align}\nonumber
L^{(1)} =&\left(
    \begin{array}{cccc}
    -3  &   3  &   0 &  0\\
     3  &  -6  &   2 &  1\\
     0  &   2  &  -4 &  2\\
     0  &   1  &   2 &  -3
    \end{array}
  \right),\\\nonumber
  L^{(2)} =&\left(
    \begin{array}{ccc}
    -1  &  1  &  0\\
    1   &  -2  &  1\\
    0  &  1  &  -1
    \end{array}
  \right),\\
  L^{(3)} =&\left(
    \begin{array}{cc}
    -2  &   2  \\
     2  &  -2
    \end{array}
  \right).
\end{align}
The largest nonzero eigenvalue of $L^{(1)}, L^{(2)}$ and $L^{(3)}$ are $\lambda_2^{(1)}=-2.0905$, $\lambda_2^{(2)}=-1$, $\lambda_2^{(3)}=-4$.

{\bf Example 1}. In this example, the initial state of each node in the network is randomly chosen in $[0,1]$. The coupling coefficients are $c_1=1$, $c_2=1$ and $c_3=1$.

\begin{figure}
\centering
\includegraphics[width=0.4\textwidth]{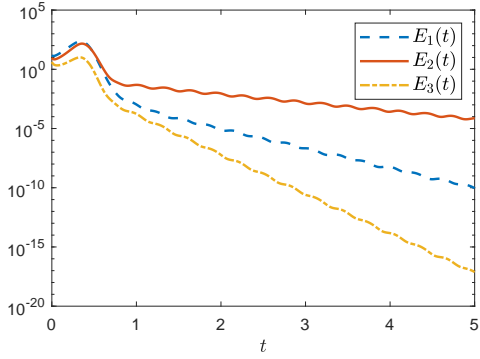}
\caption{Dynamics of $E_{1}(t)$, $E_{2}(t)$ and $E_3(t)$ with respect to time $t$ for the linearly coupled high order layered network (\ref{model}). }
\label{fig1}
\end{figure}

In Fig. \ref{fig1}, we plot the dynamics of $E_{1}(t)$, $E_{2}(t)$ and $E_3(t)$ with respect to time, which shows that $E_3(t)$ converges to zero fastest and $E_2(t)$ converges slowest. This numerical result coincides with the fact that $c_3\lambda_2^{(3)}\le c_1\lambda_2^{(1)} \le c_2\lambda_2^{(2)}$.

\begin{figure}[htbp]
\centering
    \begin{minipage}[t]{0.8\linewidth}
        \centering
        \includegraphics[width=\textwidth]{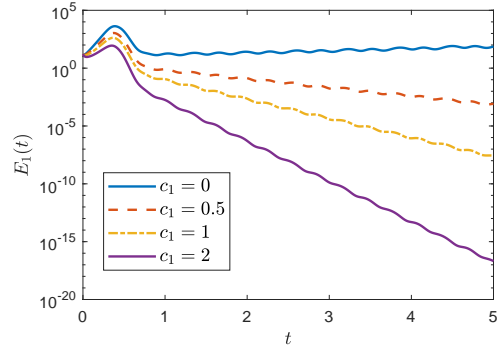}
        \centerline{\footnotesize{(a) }}
    \end{minipage}\\
    \begin{minipage}[t]{0.8\linewidth}
        \centering
        \includegraphics[width=\textwidth]{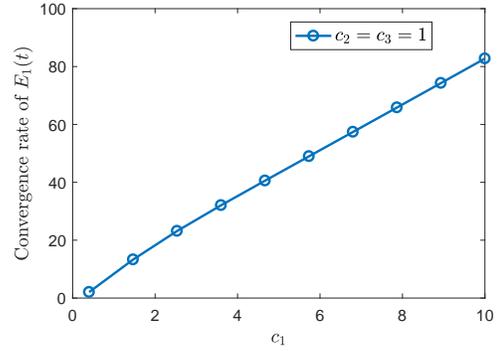}
        \centerline{\footnotesize{(b)}}
    \end{minipage}
    \caption{(a) Dynamics of $E_{1}(t)$ under different values of $c_1$. (b) The convergence rate of $E_1(t)$ with respect to $c_1$. The parameters that are kept fixed are $c_2=1$, $c_3=1$.}
\label{fig2}
\end{figure}

{\bf Example 2}. In this example, we want to reveal the relevance between synchronization and coupling coefficients. For brevity, we only present the results on coupling coefficient $c_1$ since the results on the rest two coefficients are similar.

From Fig. \ref{fig2}, it can be seen that the larger that $c_1$ is, the faster $E_1(t)$ converges to zero. The convergence rate of $E_1(t)$ scales linearly with respect to $c_1$, which coincides with the estimation (\ref{1th-syn}). That is, $1$st order synchronization depends on the coupling coefficient $c_1$.

Dynamics of $E_2(t)$ and $E_3(t)$ give different results.
Here we only present the results on $E_2(t)$ since dynamics of $E_2(t)$ and $E_3(t)$ give the similar results. From Fig. \ref{fig3}, it can be seen that as $c_1$ increases, the convergence rate of $E_2(t)$ increases distinctly at first, then rarely changes for $c_1>2$. This is because strengthening interactions among nodes in the $1$st order subnetwork enhances the $2$nd order synchronization. But the enhancement is limited since the $2$th order synchronization is mainly governed by the value of $c_2$ and $\lambda_2^{(2)}$.

\begin{figure}[htbp]
\centering
    \begin{minipage}[t]{0.8\linewidth}
        \centering
        \includegraphics[width=\textwidth]{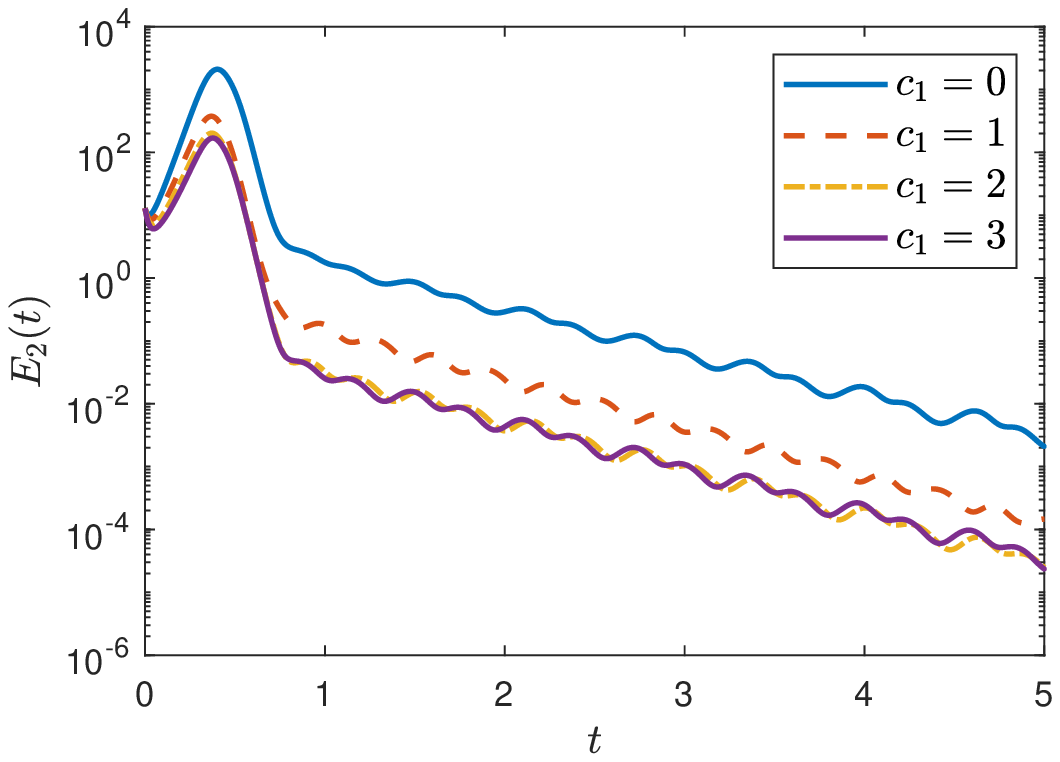}
        \centerline{\footnotesize{(a) }}
    \end{minipage}\\
    \begin{minipage}[t]{0.8\linewidth}
        \centering
        \includegraphics[width=\textwidth]{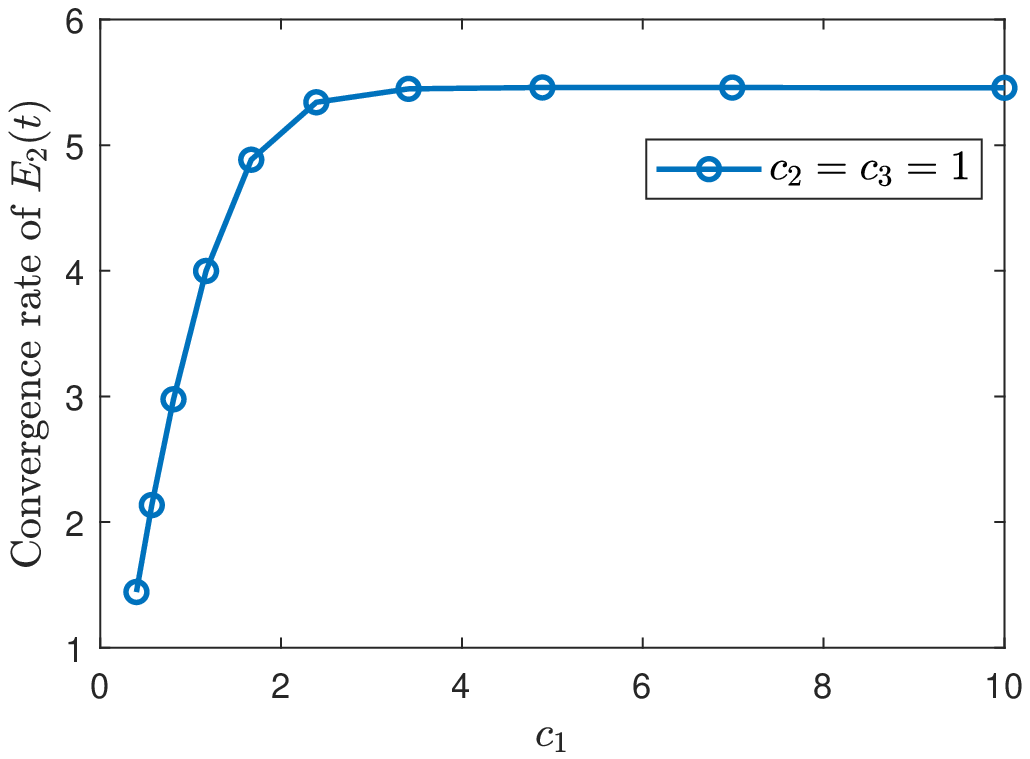}
        \centerline{\footnotesize{(b)}}
    \end{minipage}
    \caption{(a) Dynamics of $E_{2}(t)$ under different values of $c_1$. (b) The convergence rate of $E_2(t)$ with respect to $c_1$. The parameters that are kept fixed are $c_2=1$, $c_3=1$. }
\label{fig3}
\end{figure}
\section{Conclusions}\label{sec-conc}
In this paper, we investigated synchronization of the high order layered network (\ref{model}) and explored relations between synchronization and the eigenvalues of the coupling matrices. An effective theoretical analysis was provided. Numerical examples confirmed our analytic results.


\begin{thebibliography}{99}
\bibitem{Pecora1990}
L. M. Pecora and T. L. Carroll, ``Synchronization in chaotic systems", {\it Physical Review Letters}, vol. 64, no. 8, pp. 821-824, 1990.

\bibitem{Pecora2000}
L. Pecora, T. Carroll, G. Johnson, D. Mar, and K. S. Fink,
``Synchronization stability in coupled oscillator arrays: solution for arbitrary configuration", {\it International Journal of Bifurcation and Chaos}, vol. 10, no. 2, pp. 273-290, 2000.

\bibitem{Wu1995}
C. W. Wu and L. O. Chua, `` Synchronization in an array of linearly coupled dynamical systems", {\it IEEE Transacions on Circuits and Systems I: Fundamental Theory and Applications }, vol. 42, no. 8, pp. 430-447, 1995.

\bibitem{Wu2005}
C. W. Wu, ``Synchronization in networks of nonlinear dynamical systems coupled via a directed graph", {\it Nonlinearity}, vol. 18, no. 3, pp. 1057-1064, 2005.


\bibitem{Chen06New}
 W. Lu, T. Chen, ``New approach to synchronization analysis of linearly coupled ordinary differential systems", {\it Physica D: Nonlinear Phenomena}, vol. 213, no. 2, pp. 214-230, 2006.

\bibitem{Lu2021QUAD}
W. Lu and T. Chen, QUAD-condition, synchronization, consensus of multiagents, and anti-synchronization of complex networks, {\it IEEE Transactions on Cybernetics}, vol. 51, no. 6, pp. 3384-3388, June 2021.

\bibitem{Delellis2011On}
P. Delellis, M. D. Bernardo, G. Russo, ``On QUAD, Lipschitz, and contracting vector fields for consensus and synchronization of networks", {\it IEEE Transacions on Circuits and Systems 1: Regular Papers}, vol. 58, no. 3, pp.576-583, 2011.

\bibitem{Gao2012Networks}
J. Gao, S. V. Buldyrev, H. E. Stanley, and S. Havlin, ``Networks formed from interdependent networks", {\it Nature Physics}, vol. 8, pp. 40-48, 2012.

\bibitem{Wu2020Impact}
D. Wu, M. Tang, Z. Liu, and Y. Lai, ``Impact of inter-layer hopping on epidemic spreading in a multilayer network", {\it Communications in Nonlinear Science and Numerical Simulation}, vol. 90, pp. 105403, 2020.



\bibitem{Sevilla2016Interlayer}
R. Sevilla-Escoboza, I. Sendi\~na-Nadal, I. Leyva, R. Guti\'errez, J. M. Buld\'u, S. Boccaletti, Inter-layer synchronization in multiplex networks of identical layers, {\it Chaos}, vol. 26, 065304, 2016.



\bibitem{Deng2020Eigenvalue}
Y. Deng, Z. Jia, G. Deng, and Q. Zhang, ``Eigenvalue spectrum and synchronizability of multiplex chain networks", {\it Physica A: Statistical Mechanics and its Applications}, vol. 537, 122631, 2020.

\bibitem{Khalaf2019Synchronization}
A. J. M. Khalaf, F. E. Alsaadi, F. E. Alsaadi, V. Pham, and K. Rajagopal, ``Synchronization in a multiplex network of gene oscillators", {\it Physics Letters A}, vol. 383, no. 31, 125919, 2019.

\bibitem{Bidesh2019Intralayer}
B. K. Bera, S. Rakshit, D. Ghosh, ``Intralayer synchronization in neuronal multiplex network", {\it The European Physical Journal Special Topics}, vol. 228, 2441–2454, 2019.

\bibitem{Tang2019Master}
L. Tang, X. Wu, J. L{\"u}, J. Lu, and R. M. D'Souza, ``Master stability functions for complete, intra-layer and inter-layer synchronization in multiplex networks," {\it Physical Review E}, vol. 99, 012304, 2019.



\bibitem{Gomez2013Diffusions}
S. G\'omez, A. D\'{\i}az-Guilera, J. G\'omez-Garde\~nes, C. J. P\'erez-Vicente, Y. Moreno, and A. Arenas, ``Diffusion dynamics on multiplex networks", {\it Physical Review Letters}, vol. 110, pp. 028701, 2013.


\bibitem{Xu2015Synchronizability}
M. Xu, J. Zhou, J. Lu, and X. Wu, ``Synchronizability of two-layer networks", {\it The European Physical Journal B}, vol. 88, 240, 2015.

\end{thebibliography}

\end{document}